\documentclass[11pt]{article}
\usepackage{amsmath}
\usepackage{amssymb,amsbsy,amsthm}
\usepackage{graphicx}
\usepackage{enumerate}
\usepackage{dsfont}
\usepackage[margin=1in]{geometry}
\usepackage{pdfsync}
\usepackage{hyperref}
\usepackage{algpseudocode}

\allowdisplaybreaks[1]

\newcommand{\eps}{\varepsilon} 

\newcommand{\ed}{\mathrm{ed}} 
\newcommand{\edd}{\underline{\mathrm{ed}}} 
\newcommand{\wh}{\widehat}  
\newcommand{\wt}{\widetilde} 

\newcommand{\cov}{\mathrm{cov}} 
\newcommand{\ind}{\mathrm{ind}} 

\newcommand{\LCS}{\mathrm{LCS}} 

\newtheorem{theorem}{Theorem}[section]
\newtheorem{lemma}[theorem]{Lemma}

\newtheorem{corollary}[theorem]{Corollary}

\begin{document}

\title{Sequence assembly from corrupted shotgun reads}
\author{
	Shirshendu Ganguly
	\thanks{University of Washington; \texttt{sganguly@math.washington.edu}.}
	\and
	Elchanan Mossel
	\thanks{University of Pennsylvania and University of California, Berkeley; \texttt{mossel@wharton.upenn.edu}.}
	\and 
	Mikl\'os Z.\ R\'acz
	\thanks{Microsoft Research; \texttt{miracz@microsoft.com}.}
}
\date{\today}

\maketitle


\begin{abstract}
The prevalent technique for DNA sequencing consists of two main steps:
shotgun sequencing,  
where many randomly located fragments, called reads, are extracted from the overall sequence, 
followed by an assembly algorithm that aims to reconstruct the original sequence. 
There are many different technologies that generate the reads: 
widely-used second-generation methods create short reads with low error rates, 
while emerging third-generation methods create long reads with high error rates. 
Both error \emph{rates} and error \emph{profiles} differ among methods, 
so reconstruction algorithms are often tailored to specific shotgun sequencing technologies. 
As these methods change over time, 
a fundamental question is whether there exist reconstruction algorithms which are \emph{robust}, 
i.e., which perform well under a wide range of error distributions.

Here we study this question of sequence assembly from corrupted reads. 
We make no assumption on the \emph{types} of errors in the reads, 
but only assume a bound on their \emph{magnitude}.  
More precisely, for each read we assume that 
instead of receiving the true read with no errors, 
we receive a corrupted read which has edit distance at most 
$\varepsilon$ times the length of the read 
from the true read. 
We show that if the reads are long enough and there are sufficiently many of them, 
then approximate reconstruction is possible: 
we construct a simple algorithm such that 
for almost all original sequences the output of the algorithm is 
a sequence whose edit distance from the original one is at most 
$O(\varepsilon)$ times the length of the original sequence. 
\end{abstract}


\section{Introduction} \label{sec:intro} 

DNA sequencing is by now an essential element of a variety of biological and clinical studies. 
Current de novo sequencing technologies typically have two main stages. 
First, many randomly located fragments, called reads, are extracted from the DNA sequence in a process called shotgun sequencing. 
Next, an assembly algorithm aims to reconstruct the original sequence based on overlaps between the reads.

The rise of second-generation sequencing methods, such as Illumina, have resulted in many advances in the past decade because they generate high-throughput data cheaply and quickly. 
However, these methods produce short reads (a few hundred basepairs long) in order to have a low error rate ($1$-$3\%$), which results in incomplete and fragmented assemblies~\cite{Salzberg:10}. 
Emerging technologies, such as PacBio's Single Molecule Real-Time sequencing technology and Oxford Nanopore Technologies, were developed in part to solve this problem. 
They produce long reads (over ten thousand basepairs long), but currently suffer from a high error rate ($10$-$22\%$) (see, e.g.,~\cite{Chin_et_al:13,LoQuSi:15,Li:15}).

Not only do these different shotgun sequencing methods produce reads with different error \emph{rates}, 
they also have different error \emph{profiles}, 
e.g., due to various systematic errors. 
Consequently assembly algorithms are often tailored to specific sequencing technologies 
to exploit their unique properties. 
As these technologies will inevitably change and new ones will arise, 
a fundamental question is the \emph{robustness} of reconstruction algorithms. 
Will the current ones still be useful a decade from now? 
Are there algorithms which perform well under a wide variety of error distributions? 
This is the question we study in this paper.

Several recent papers 
have taken an \emph{information-theoretic} point of view 
to the sequence assembly problem. 
The basic question is: what are the fundamental limits to \emph{any} assembly algorithm? 
Given a sequencing technology and the statistics of the DNA sequence, 
how long do the reads need to be and 
how many are required for reconstruction? 
Motahari, Bresler, and Tse~\cite{MoBrTs:13} study this question assuming an i.i.d.\ DNA sequence and error-free reads, 
and show a sharp phase transition: 
if the reads are short enough to have repeats, then reconstruction is impossible, 
while as long as the reads are long enough to have no repeats, 
the necessary condition of having enough reads to cover the whole DNA sequence is essentially sufficient. 
Previously Dyer, Frieze, and Suen~\cite{DyFrSu:94} obtained the same phase transition in the length of the reads 
assuming sequencing by hybridization, i.e., that a copy of every read is available. 
Several extensions and variations of this problem have been studied. 
Adding some amount of i.i.d.\ noise to the reads still allows for reconstruction of the perfect layout of the reads~\cite{MoRaTsMa:13}. 
In~\cite{BrBrTs:13} the authors give a sufficient condition for reconstruction for any sequence based on its repeat statistics, assuming error-free reads. 
This was later extended to allow the reads to come from an erasure error model~\cite{ShCoTs:15}. 
These papers are discussed in more detail later.

In this paper we continue this line of work, assuming an i.i.d.\ DNA sequence as in~\cite{MoBrTs:13,MoRaTsMa:13}. 
The main novelty in the model we consider is a strong \emph{adversarial corruption/error model} on the reads. 
More precisely, for each read we assume that instead of receiving the true read with no errors, 
we receive a corrupted read with the edit distance between the true and the corrupted reads being at most $\varepsilon$ times the length of the true read. 
Given such a strong adversarial error model, we relax our goal from perfect reconstruction to approximate reconstruction. 
Our main contribution is to show that if the reads are long enough and there are sufficiently many of them, 
then approximate reconstruction is possible: 
we present a simple sequential algorithm for which 
the edit distance between 
the original sequence and 
the output of the algorithm 
is at most 
$O(\varepsilon)$ times the length of the original sequence.

\section{Problem setting} \label{sec:problem} 

We are interested in approximately recovering a long sequence of interest 
from a set of randomly chosen shorter reads which are arbitrarily corrupted up to a certain extent. 
Consequently the problem has four main parameters: 
sequence length $n$, 
read length $L$, 
number of reads $N$, and 
error/corruption rate $\eps$; 
a fifth parameter, $\delta$, measures the probability of unsuccessful approximate reconstruction.

Before defining the problem precisely, we introduce some notation. 
Let $\Sigma$ be a finite alphabet from which the entries of the sequence come from; 
in the case of DNA sequencing we have $\Sigma = \left\{ A, C, G, T \right\}$. 
For a sequence $x = \left( x_1, x_2, \dots, x_n \right) \in \Sigma^n$ 
and integers $i$ and $j$, let $x\left[i,j\right]$ denote the substring $\left( x_i, x_{i+1}, \dots, x_j \right)$. 
Let $\Sigma^{*} = \cup_n \Sigma^n$. 
For $x,y \in \Sigma^{*}$, 
let $\ed \left( x, y \right)$ denote the edit distance between $x$ and $y$, 
i.e., the minimum number of deletion, insertion, or substitution operations necessary to go from $x$ to $y$.

The approximate reconstruction problem with parameters $\left( n, L, N, \eps, \delta \right)$ is then defined as follows; 
see Fig.~\ref{fig:problem} for an illustration. 
\begin{itemize}
 \item The sequence of interest, $X \in \Sigma^n$, is chosen uniformly at random among all possible sequences; 
  i.e., the entries of $X$ are i.i.d.\ chosen uniformly from the alphabet $\Sigma$.\footnote{We focus on the uniform distribution for simplicity; extensions to certain other distributions should also be tractable.} 
 \item The data are corrupted reads of $X$, defined as follows. 
  Let $\left\{ T_i \right\}_{i=1}^{N}$ be i.i.d.\ uniform in $\left\{1,2, \dots, n - L + 1 \right\}$ (the starting positions of the reads), 
  and let $R_i = X \left[ T_i, T_i + L - 1 \right]$ be the $i^{\text{th}}$ (uncorrupted) read. 
  Instead of receiving the (multi)set of (uncorrupted) reads $\mathcal{R} = \left\{ R_1, R_2, \dots R_N \right\}$, we receive a (multi)set of corrupted reads 
  \[
    \wt{\mathcal{R}} = \left\{ \wt{R}_1, \wt{R}_2, \dots, \wt{R}_N \right\},
  \]
  where the only thing we know is that 
  \begin{equation}\label{eq:error_rate}
    \ed \left( R_i, \wt{R}_i \right) \leq \eps L, \qquad \text{ for every } i \in \left[N\right],
  \end{equation}
  but otherwise the $\wt{R}_i$ can be arbitrary.\footnote{The choice of edit distance arises naturally from DNA sequencing; one can consider the problem with other notions of distance as well.}
 \item The goal of an approximate reconstruction algorithm is to output a sequence $\wh{X} \in \Sigma^{*}$ such that 
    \begin{equation}\label{eq:goal}
      \ed \left( X, \wh{X} \right) \leq C \eps n
    \end{equation}
  for some absolute constant $C$, with probability at least $1 - \delta$ (for all $n$ large enough).
\end{itemize}

\begin{figure}[h!]
  \centering
  \includegraphics[width=\linewidth]{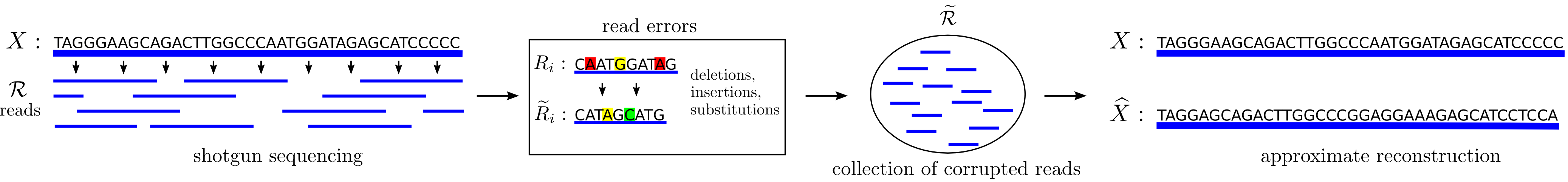}
  \caption{A schematic of the approximate reconstruction problem.}
  \label{fig:problem}
\end{figure}

When $\eps = 0$ (i.e., there are no errors in the reads), this amounts to exact reconstruction of the original sequence, 
which was studied and solved in~\cite{MoBrTs:13}. 

If an algorithm achieves~\eqref{eq:goal} for a given error rate $\eps > 0$ with a given constant $C$, 
then we say that it is an \emph{approximate reconstruction algorithm for error rate $\eps$ with approximation factor $C$}. 
Note that for the empty string $\emptyset$ we have $\ed \left( X, \emptyset \right) = n$, 
so for a given approximation factor $C$ the problem is interesting only for $\eps < 1 / C$; 
of course the goal is minimize the approximation factor $C$.

\section{Results} \label{sec:results} 

Before presenting our results on the approximate reconstruction problem described above, 
we first recall the main results of~\cite{MoBrTs:13} characterizing the case of $\eps = 0$, i.e., error-free reads. 
In both cases the interesting regime is when the read length $L$ scales as the logarithm of the sequence length $n$, 
so in the following we let 
$L = \overline{L} \ln \left( n \right)$, 
where $\overline{L}$ is constant.\footnote{For the sake of readability, we refrain from rounding non-integer values that are meant to be integers, such as $\overline{L} \ln \left( n \right)$.}

When there are no errors in the reads, 
then there are two obstructions to reconstruction.
First, if the reads are too short, 
then there will be repeats coming from different parts of the original sequence, 
which create ambiguity in reconstruction, 
even if all substrings of length $L$ of the sequence are given. 
This observation goes back to the work of Ukkonen~\cite{Ukkonen:92}, 
who characterized the patterns that preclude exact reconstruction of the sequence. 
This is often referred to as the 
\emph{repeat-limited regime}.

Second, even if $L$ is large enough, it is necessary to have enough reads to cover the original sequence; 
otherwise the data does not contain enough information for exact reconstruction. 
Define 
$N_{\cov} = N_{\cov} \left( n, L, \delta \right)$ 
as the minimum number of reads necessary such that with probability at least $1 - \delta$ 
the randomly located reads cover the entire original sequence. 
Lander and Waterman~\cite{LanderWaterman:88} first studied the coverage properties of shotgun sequencing 
and showed that 
$N_{\cov} \approx \frac{n}{L} \ln \left( \frac{n}{L \delta} \right)$. 
This is often referred to as the 
\emph{coverage-limited regime}.

In the main result of~\cite{MoBrTs:13}, Motahari, Bresler, and Tse show that when there are no errors in the reads, 
these are the only two obstructions to reconstruction. 

\begin{theorem}[Exact reconstruction from error-free reads~\cite{MoBrTs:13}] \label{thm:MBT13}
  If $\overline{L} < 2 / \ln \left( \left| \Sigma \right| \right)$, then no algorithm can reconstruct the original sequence exactly with probability greater than $1/2 + o(1)$ (as $n \to \infty$). 

  If $\overline{L} > 2 / \ln \left( \left| \Sigma \right| \right)$, then exact reconstruction is possible, 
and the necessary condition of coverage is essentially sufficient. 
  More precisely, let $N_{\min} \left( n, L, \delta \right)$ denote the minimum number of reads needed to reconstruct the original sequence exactly with probability at least $1 - \delta$. 
  If $\overline{L} > 2 / \ln \left( \left| \Sigma \right| \right)$, then for every fixed $\delta \in (0,1/2)$, 
  \[
    \lim_{n \to \infty, L = \overline{L} \ln \left(n\right)} \frac{N_{\min} \left( n, L, \delta \right)}{N_{\cov} \left( n, L, \delta \right)} = 1.
  \]
\end{theorem}

When the reads are corrupted, the two obstructions to reconstruction discussed above remain. 
In fact, the repeat-limited regime is slightly larger in the presence of corruption. 
If 
$\overline{L} < 2 / \left\{  \left( 1 - \eps \right) \ln \left( \left| \Sigma \right| \right) \right\}$ 
then consider the corruption process which deletes the last $\eps L$ coordinates of every read. 
The reads then have normalized length $\left( 1 - \eps \right) \overline{L} < 2 / \ln \left( \left| \Sigma \right| \right)$, 
so by Theorem~\ref{thm:MBT13} exact reconstruction is impossible.

The picture describing successful algorithms is less clean when the reads are corrupted. 
This is primarily due to the desideratum of approximate reconstruction. 
When there are no errors, the success of an algorithm is binary: either it reconstructs the original sequence exactly or it does not. 
Here, however, an algorithm can have varying degrees of success based on the approximation factor $C$ that it achieves in~\eqref{eq:goal}. 
Even for those parameters $\left( \overline{L}, N \right)$ for which an algorithm with finite approximation factor $C$ exists, the best achievable $C$ might depend on $\left( \overline{L}, N \right)$.

With this in mind, our goal in this paper is to show that 
when $\overline{L}$ and $N$ are large enough, 
then there exists an approximate reconstruction algorithm with finite approximation factor.

\begin{theorem}[Approximate reconstruction from corrupted reads]\label{thm:main_eps}
  There exist constants $C$ and $\overline{C}$, depending only on $\left| \Sigma \right|$, such that 
  for every $\eps > 0$, 
  if $\overline{L} > \overline{C} / \eps$ and $N > N_{\cov} / \eps$, 
  then there exists an approximate reconstruction algorithm for error rate $\eps$ with approximation factor $C$. 
\end{theorem}

We show that a simple sequential algorithm achieves this result. 
Starting with an arbitrary read, 
at each step of the algorithm we find a read that overlaps with the current partially reconstructed sequence 
and extend the sequence using this read. 
If $\overline{L}$ and $N$ are large enough, 
then in each step of this process we extend the sequence by at least $cL$ for some positive constant $c$, 
while incurring an error in edit distance of at most $O(\eps) L$. 
Estimates on the edit distance between random strings with some overlap are crucially used in the analysis. 
The algorithm terminates when it has approximately reached both ends of the original sequence, 
and results in an estimate with the guarantee given by Theorem~\ref{thm:main_eps}.

Several variants of such an algorithm can be considered, 
and it can be shown using one of them that 
the dependence of $\overline{L}$ and $N$ on $\eps$ in Theorem~\ref{thm:main_eps} is not necessary. 
However, we decided to focus on one particular variant 
because it results in a small approximation factor (close to $3$) 
when $\eps$ is small enough, as stated in the following theorem.

\begin{theorem}[Approximate reconstruction from corrupted reads]\label{thm:main_eps_3}
  For every $C > 3$ 
  there exist constants $\overline{C} = \overline{C} \left( \Sigma \right)$, $\eps_0 = \eps_0 \left( \Sigma, C \right)$ and $C' = C' \left( \Sigma, C \right)$ 
  such that for every $\eps\in \left( 0, \eps_0 \right)$ 
  if $\overline{L} \geq \overline{C} / \eps$ and $N \geq C' N_{\cov} / \eps$, 
  then there exists an approximate reconstruction algorithm for error rate $\eps$ with approximation factor $C$. 
\end{theorem}

The closest results to Theorems~\ref{thm:main_eps} and~\ref{thm:main_eps_3} in the literature are those of~\cite{MoRaTsMa:13} and~\cite{ShCoTs:15}. 
In~\cite{MoRaTsMa:13} the authors consider i.i.d.\ noise affecting the reads of an i.i.d.\ sequence, 
and show that perfect layout (where all the reads are mapped correctly to their true locations) is possible even 
when the noise level is relatively high. 
The main reason for this positive result is that the independent noise assumption allows error correction of the reads by averaging; 
in the adversarial error model considered here such averaging is not always possible, hence the weaker goal of approximate reconstruction and the results of Theorems~\ref{thm:main_eps} and~\ref{thm:main_eps_3}.

In the recent follow-up work~\cite{ShCoTs:15}, the authors show positive results for a more realistic adversarial error model. 
They consider arbitrary sequences and give a bound on the read length, as a function of the repeat statistics of the sequence and the error rate, above which perfect assembly is possible. 
However, the model they consider simplifies several aspects of the problem. 
The authors mention several possible extensions as avenues for future work, two of which we consider in this paper:
more general errors, and a shotgun read model. 
For one, they specifically consider \emph{erasure errors}, where symbols in a read are erased, but the locations of the erasures are known. 
Furthermore they assume bounds not only on the number of erasures in a read, but also on the number of reads in which a given base is erased. 
This means that the reads contain information about the whole sequence. 
The error model we consider is much more adversarial, 
e.g., it can happen that even all the reads together contain no information at all about $\eps n$ bases of the sequence due to deletions. 
Also, they consider a dense-read model, where all reads of length $L$ of the original sequence are provided, 
therefore bypassing the question of coverage depth necessary for assembly. 
Here we instead consider the more realistic shotgun read model and provide a sufficient bridging condition for approximate reconstruction.

In summary, while our reconstruction results are weaker than previous ones, 
this is due to the much stricter adversarial error model we consider. 
Going forward, the main challenge is to bridge the gap between these models and results.

\section{Sequential reconstruction algorithm and its analysis} \label{sec:aofa} 

We first present some results on the edit distance between random strings in Section~\ref{sec:edit} 
that then allow us to present a simple sequential reconstruction algorithm in Section~\ref{sec:algo}. 
This algorithm is then analyzed and shown to have the desired performance in Section~\ref{sec:analysis}.

\subsection{Results on the edit distance between random strings} \label{sec:edit} 

Since the only information we have about the corrupted reads is that their edit distance from the actual reads is not too large (see~\eqref{eq:error_rate}), 
it is essential for any reconstruction algorithm to have a good understanding of and to make use of the edit distance between pairs of reads. 
Accordingly, we now present results on the edit distance between random strings with overlap; 
the proofs of these results are in Appendix~\ref{sec:edit_distance_proofs}.

Consider two random strings with overlap.  
Simulations show (see Figure~\ref{fig:edit}) 
that there is a phase transition in the edit distance between the two as a function of the overlap. 
If the overlap is above a certain threshold (which is linear in the length of the strings), 
then the edit distance is exactly twice the length of the overhang; 
if the overlap is below this threshold, 
then the edit distance between the two strings is as if they were completely independent. 
The results below rigorously verify certain aspects of this picture.

\begin{figure}[h!]
  \centering
  \includegraphics[width=0.45\linewidth]{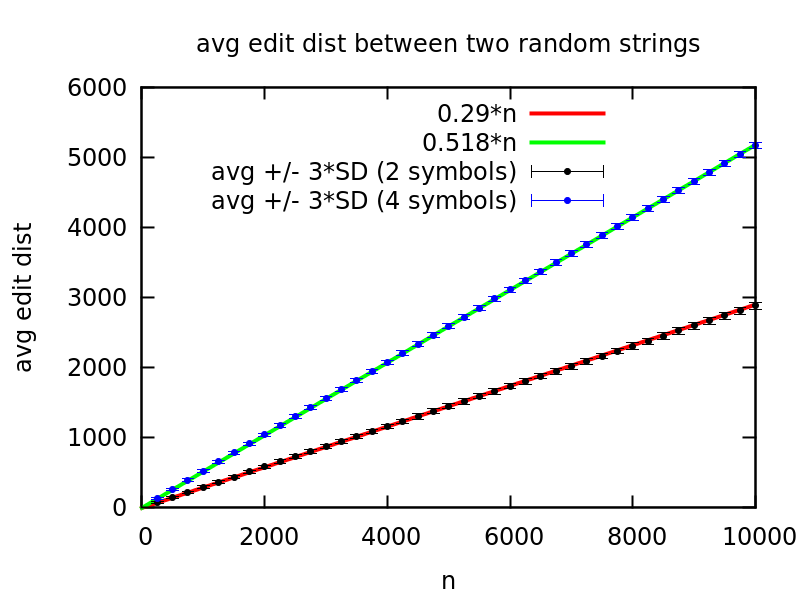}
  \qquad  
  \includegraphics[width=0.45\linewidth]{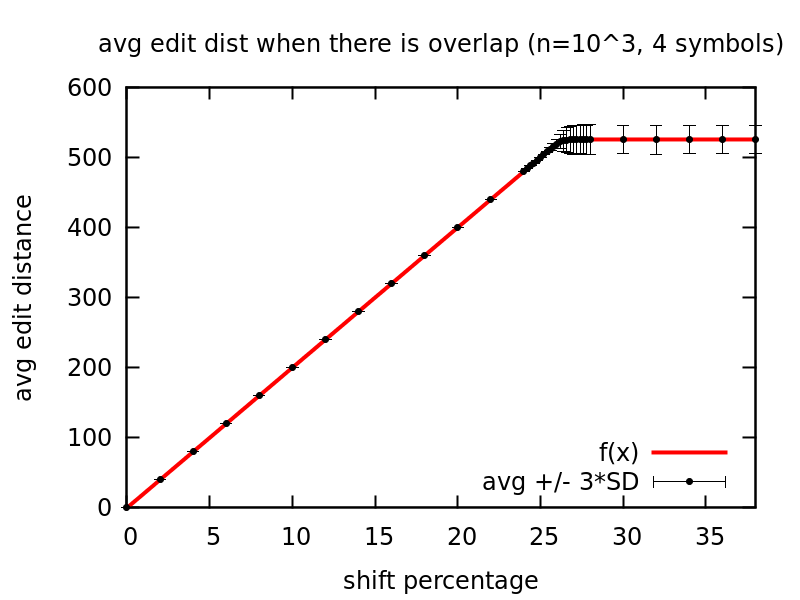}
  \caption{The left plot shows the empirical average edit distance between two independent strings of length $n$, where $n$ ranges from $250$ to $10^4$, and $\left|\Sigma\right|$ is $2$ (black) or $4$ (blue). 
  The lines $0.29 n$ (for $\left|\Sigma\right| = 2$) and $0.518 n$ (for $\left| \Sigma \right| = 4$) show a good fit to the data, though the limiting slope appears to be somewhat smaller than $0.29$ and $0.518$, respectively. 
  The right plot shows the empirical average edit distance between two strings of length $n = 10^3$ and with $\left| \Sigma \right| = 4$, where one string is a shift of the other, as in the setup of Lemma~\ref{lem:large_overlap}. 
  The function $f$ is piecewise linear with two pieces: 
  it is equal to twice the length of the shift when the shift percentage is less than $26.25\%$, 
  and it is equal to the constant $525$ ($52.5\%$ of the length $n = 10^3$) when the shift percentage is more than $26.25\%$.
  In both plots the average is over $10^3$ runs and the error bars show plus/minus $3$ times the empirical standard deviation.}
  \label{fig:edit}
\end{figure}

\begin{lemma}\label{lem:indpt}
  Let $X_m, Y_m \in \Sigma^m$ be two independent uniformly random strings. 
  There exists an absolute constant $c_{\ind} = c_{\ind} \left( \Sigma \right) > 0$ such that almost surely
  \[
   \lim_{m \to \infty} \frac{1}{m} \ed \left( X_m, Y_m \right) = c_{\ind}.
  \]
\end{lemma}

Determining the value of the limiting constant $c_{\ind}$ is a challenging open problem. 
When $\left| \Sigma \right| = 4$, as in the case of DNA sequencing, 
simulations suggest that $c_{\ind} \approx 0.51$, 
while we show a simple lower bound of $c_{\ind} > 0.338$.

\begin{lemma}\label{lem:small_overlap}
  Let $X \in \Sigma^{2m}$ be a uniformly random string. 
  For every $d \in \left( 0, 1 \right)$ there exist  positive constants $\gamma = \gamma \left( d, \Sigma \right)$ and $c' = c' \left( \Sigma \right)$ such that 
  \[
    \ed \left( X \left[ 1 , m \right], X \left[k + 1, k + m  \right] \right) \geq \gamma m
  \]
  for all $k \ge d m$ with probability at least $1 - e^{-c' d m}$. 
\end{lemma}

In particular, when $\left| \Sigma \right| = 4$, then we can take 
$\gamma = 0.338 \times d$. 

\begin{lemma}\label{lem:large_overlap}
  Let $X \in \Sigma^{2m}$ be a uniformly random string. 
  There exist positive constants $c = c \left( \Sigma \right)$ and $c' = c' \left( \Sigma \right)$ such that 
  \[
    \ed \left( X \left[ 1 , m \right], X \left[ 1 + k , m + k \right] \right) = 2k
  \]
  for all $k \leq c m$ with probability at least $1 - e^{-c' m}$. 
\end{lemma}

We denote by $\kappa_{\ed} = \kappa_{\ed} \left( \Sigma \right)$ the supremum of all constants $c = c \left( \Sigma \right)$ for which Lemma~\ref{lem:large_overlap} holds (with some $c' = c' \left( \Sigma \right)$). 
Figure~\ref{fig:edit} suggests that $\kappa_{\ed} = c_{\ind} / 2$.  
We show that for $\left| \Sigma \right| = 4$, $\kappa_{\ed} > 0.0846$.

As a corollary of the lemmas we obtain the following result. 
\begin{corollary}\label{cor:shift} 
  Let $X \in \Sigma^n$ be uniformly random and let $L = \overline{L} \ln \left( n \right)$. 
  If the constant $\overline{L} = \overline{L} \left( \Sigma \right)$ is large enough, 
  then there exists a positive constant $c = c \left( \Sigma \right)$ such that 
  with probability going to $1$ as $n \to \infty$ 
  the following holds for all $i,j \in \left[ n - L + 1 \right]$: 
  \begin{enumerate}[(a)]
   \item if $\left| i - j \right| \leq c L$, then $\ed \left( X \left[ i, i + L - 1 \right], X \left[ j, j + L - 1 \right] \right) = 2 \left| i - j \right|$;
   \item otherwise $\ed \left( X \left[ i, i + L - 1 \right], X \left[ j, j + L - 1 \right] \right) \geq 2 c L$.
  \end{enumerate}
\end{corollary}

\begin{proof}
This follows directly from Lemmas~\ref{lem:small_overlap} and~\ref{lem:large_overlap}, 
together with a union bound over all possible pairs $i, j \in \left[ n - L + 1 \right]$. 
The constant $\overline{L}$ needs to be chosen large enough so that the error probabilities in Lemmas~\ref{lem:small_overlap} and~\ref{lem:large_overlap} are $o \left( n^{-2} \right)$. 
\end{proof}

\subsection{Sequential reconstruction algorithm} \label{sec:algo} 

We present now a simple sequential approximate reconstruction algorithm. 
The algorithm takes as a parameter the number of reads; 
this guarantees a certain amount of coverage. 
Let $N = C' N_{\cov} / \eps$ and $c' = 1/C'$, where we assume that $C' \geq 1$; 
standard results~\cite{LanderWaterman:88} imply that then with probability at least $1 - \delta / 2$ 
there is no gap greater than $1.1\times c' \eps L$ in between subsequent starting points of the (yet uncorrupted) reads. 
We fix $\alpha = 4 / c \left( \Sigma \right)$, where $c \left( \Sigma \right)$ is given by Corollary~\ref{cor:shift}. 
For the purposes of Theorems~\ref{thm:main_eps} and~\ref{thm:main_eps_3} we may and will assume that $\eps$ is small enough; 
in particular we assume that $\eps \leq 1 / \left( 3 \alpha \right)$.

Before specifying the algorithm we introduce further notation. 
Let negative integers denote counting coordinates from the opposite end of a sequence, 
e.g., for $x \in \Sigma^m$, $x \left[ - k, - 1 \right]$ denotes the suffix of $x$ of length $k$. 
Furthermore let $I : \wt{\mathcal{R}} \to \left[ N \right]$ denote the map that takes a corrupted read to its index, i.e., $I \left( \wt{R}_i \right) = i$.

The algorithm is as follows:
\begin{algorithmic}[1]
\State Let $Y = \wt{R}_1$ and set $k = 1$.
\While{there exists $\wt{R} \in \wt{\mathcal{R}}$ such that 
  $\ed \left( \wt{R}_k \left[ - \alpha \eps L, - 1 \right], \wt{R} \left[ 1, \alpha \eps L \right] \right) \leq \left( 2 + 2c' \right) \eps L$}\label{alg:right}
\State choose any such $\wt{R} \in \wt{\mathcal{R}}$;
\State $Y \gets$ the concatenation of $Y$ and $\wt{R} \left[ 1 + \alpha \eps L, - 1 \right]$;
\State $k \gets I \left( \wt{R} \right)$.
\EndWhile
\State Set $k = 1$. 
\While{there exists $\wt{R} \in \wt{\mathcal{R}}$ such that 
  $\ed \left( \wt{R} \left[ - \alpha \eps L, - 1 \right], \wt{R}_k \left[ 1, \alpha \eps L \right] \right) \leq \left( 2 + 2c' \right) \eps L$}\label{alg:left}
\State choose any such $\wt{R} \in \wt{\mathcal{R}}$;
\State $Y \gets$ the concatenation of $\wt{R} \left[ 1, - \alpha \eps L - 1 \right]$ and $Y$;
\State $k \gets I \left( \wt{R} \right)$.
\EndWhile
\end{algorithmic}

In words: we take an arbitrary read, extend it to the right until we possibly can (first while loop), 
and then extend it to the left until we can (second while loop).

\subsection{Analysis of the sequential algorithm} \label{sec:analysis} 

We now analyze the algorithm presented above and as a consequence prove our results:  
Theorem~\ref{thm:main_eps} follows by taking $C' = 1$, 
while Theorem~\ref{thm:main_eps_3} follows by taking $C'$ large enough.

We first recall a fact that follows directly from the dynamic programming algorithm for computing the edit distance. 
For any $m$, sequences $x, y \in \Sigma^m$, and $i,j < m$, we have
\begin{equation}\label{eq:ed_mon}
  \ed \left( x \left[ 1, i \right], y \left[ 1, j \right] \right) 
  \leq 
  \ed \left( x \left[ 1, i + 1 \right], y \left[ 1, j + 1 \right] \right).
\end{equation}
In words: deleting a coordinate from the end of both $x$ and $y$ cannot increase their edit distance.

The first thing we have to understand is the set of corrupted reads that satisfy the conditions of the while loops in 
the algorithm; 
we focus on the first while loop as the second one is analogous. 
The following lemma says that if two corrupted reads are such that the length $\alpha \eps L$ prefix of one and suffix of the other are close in edit distance, 
then the starting points of these reads are approximately $\left( 1 - \alpha \eps \right) L$ apart.

\begin{lemma}\label{lem:close_reads}
  Let $X \in \Sigma^n$ be a uniformly random string, 
  let $L = \left( \overline{C} / \eps \right) \ln \left( n \right)$, 
  and let $\wt{R}_1, \wt{R}_2 \in \wt{\mathcal{R}}$ be two corrupted reads of $X$ of length $L$.
  Suppose that 
  $\ed \left( \wt{R}_1 \left[ - \alpha \eps L, -1 \right], \wt{R}_2 \left[ 1, \alpha \eps L \right] \right) \leq \left( 2 + 2 c' \right) \eps L$. 
  If $\overline{C}$ is large enough, then with probability $1 - o \left( n^{-2} \right)$ we have that 
  \begin{equation}\label{eq:shift}
    T_2 \in \left[ T_1 + \left( 1 - \alpha \eps \right) L - \left( 2 + 2 c' \right) \eps L, T_1 + \left( 1 - \alpha \eps \right) L + \left( 2 + 2 c' \right) \eps L \right].
  \end{equation}
\end{lemma}
\begin{proof}
  Suppose that~\eqref{eq:shift} does not hold.  
  Then the overlap between $R_1 \left[ - \alpha \eps L, -1 \right]$ and $R_2 \left[ 1, \alpha \eps L \right]$ is less than 
  $\left( \alpha - \left( 2 + 2c'\right) \right) \eps L$. 
  Recall the definition of $c = c\left( \Sigma \right)$ from Corollary~\ref{cor:shift}. 
  Since $c \alpha \eps L = 4 \eps L \geq \left( 2 + 2c' \right) \eps L$, we can apply Lemmas~\ref{lem:small_overlap} and~\ref{lem:large_overlap} to get that 
  with probability $1 - o \left( n^{-2} \right)$ we have 
  $\ed \left( R_1 \left[ - \alpha \eps L, -1 \right], R_2 \left[ 1, \alpha \eps L \right] \right) \geq 2 \left( 2 + 2 c' \right) \eps L$. 
  By the triangle inequality this implies that 
  $\ed \left( \wt{R}_1 \left[ - \alpha \eps L, -1 \right], \wt{R}_2 \left[ 1, \alpha \eps L \right] \right) \geq \left( 2 + 4 c' \right) \eps L$, 
  which is a contradiction.  
\end{proof}
Since the probability that the conclusion of the lemma does not hold is $o \left( n^{-2} \right)$, 
we can take a union bound over all pairs of corrupted reads and have the conclusion of the lemma apply to all of them with probability $1 - o \left( 1 \right)$.

We are now ready to analyze the algorithm step by step. 
We show by induction that after each extension step the partially reconstructed sequence is a good approximation of 
a substring of the original sequence. 

\begin{lemma}\label{lem:analysis}
Let $X \in \Sigma^n$ be a uniformly random string, 
let $L = \left( \overline{C} / \eps \right) \ln \left( n \right)$ where $\overline{C}$ is a large enough constant, 
and let $N = C' N_{\cov} / \eps$. 
Let $\alpha = 4 / c \left( \Sigma \right)$, where $c \left( \Sigma \right)$ is given by Corollary~\ref{cor:shift}. 
Let $Y_i$ be the state of the partially reconstructed sequence $Y$ after $i$ corrupted reads have been processed by the algorithm; 
we have $Y_1 = \wt{R}_1$. 
Let $\tau_1$ be the number of reads processed in the first while loop of the algorithm, 
and let $\tau_2$ be the number of reads processed in the second while loop.  
Also let $\tau = \tau_1 + \tau_2$, the total number of reads processed during the algorithm. 
With probability at least $1 - \delta$ 
(over the choice of $X$ and the starting points of the reads in $\mathcal{R}$) 
we have the following: 
\begin{enumerate}[(a)]
 \item\label{lem:analysis_i} For every $i \leq \tau$ 
  there exist $a_i, b_i \in \left[ n \right]$ such that 
  $\left| a_i - b_i \right| \geq \left( 1 - \left( \alpha + 2 + 2c' \right) \eps \right) i L$ and 
  \begin{equation}\label{eq:Y_i}
   \ed \left( Y_i, X \left[ a_i, b_i \right] \right) \leq  \left( 3+2c' \right) \eps i L; 
  \end{equation}
 \item\label{lem:analysis_tau} $\tau \leq \frac{n}{\left( 1 - \left( \alpha + 2 + 2c' \right) \eps \right) L}$;
 \item\label{lem:analysis_X_tau} $\ed \left( X \left[ a_{\tau}, b_{\tau} \right], X \right) \leq 2 L$.
\end{enumerate}
\end{lemma}

\begin{proof}
  Part~\eqref{lem:analysis_i} of the lemma holds for $i = 1$ by choosing $a_1 = T_1$ and $b_1 = T_1 + L - 1$. 
  For larger $i$ we prove the statement by induction on $i$. 

  Suppose we are in the first while loop of the algorithm, i.e., $i \leq \tau_1$. 
  We set $a_i = a_1$ for all $i \leq \tau_1$ and only change $b_i$. 
  Let $\wt{T}_i = T \left( \wt{R}_{k \left( i \right)} \right)$, 
  where $k \left( i \right)$ is the index of the read chosen at the $i^{\text{th}}$ round, 
  and set $b_i := \wt{T}_i + L - 1$.
  As mentioned before, we may assume that there is no gap greater than $2c' \eps L$ in between subsequent starting points of the reads. 
  Therefore if $\wt{T}_i \leq n - 2L$, 
  then there must exist $R \in \mathcal{R}$ such that 
  $T \left( R \right) - \wt{T}_i \in \left[ \left( 1 - \alpha \eps \right) L - c'\eps L, \left( 1 - \alpha \eps \right) L + c'\eps L \right]$. 
  By the triangle inequality this implies that 
  $\ed \left( \wt{R}_{k \left( i \right)} \left[ -\alpha \eps L, -1 \right], \wt{R} \left[ 1, \alpha \eps L \right] \right) \leq \left( 2 + 2c' \right) \eps L$, 
  i.e., $\wt{R}$ satisfies the condition of the while loop. 
  Thus $\wt{T}_{\tau_1} > n - 2L$. 
  Now take any $\wt{R} \in \wt{\mathcal{R}}$ that satisfies the condition of the while loop. 
  By Lemma~\ref{lem:close_reads} we know that 
  $T \left( \wt{R} \right) - \wt{T}_i - \left( 1 - \alpha \eps \right) L \in \left[ - \left( 2 + 2c' \right) \eps L, \left( 2 + 2c' \right) \eps L \right]$. 
  By subadditivity and the induction hypothesis we have that 
  \begin{align*}
   \ed \left( Y_{i+1}, X \left[ a_{i+1}, b_{i+1} \right] \right) 
    &\leq 
    \ed \left( Y_i, X \left[ a_i, b_i \right] \right) + \ed \left( \wt{R} \left[ 1 + \alpha \eps L, -1 \right], X \left[ b_{i} + 1, b_{i+1} \right] \right) \\
    &\leq \left(3+2c'\right) \eps i L + \ed \left( \wt{R} \left[ 1 + \alpha \eps L, -1 \right], X \left[ b_{i} + 1, b_{i+1} \right] \right),
  \end{align*}
  so it suffices to estimate the latter term. By~\eqref{eq:ed_mon} we have that 
  $\ed \left( \wt{R} \left[ 1 + \alpha \eps L, -1 \right], R \left[ 1 + \alpha \eps L , L \right] \right) \leq \ed \left( \wt{R}, R \right) \leq \eps L$. 
  Using the definition of $b_i$ we have that 
  \[
   \ed \left( R \left[ 1 + \alpha \eps L, L \right], X \left[ b_{i} + 1, b_{i+1} \right] \right) = \left| \left( \wt{T}_i + L - 1 \right) - \left( T \left( R \right) + \alpha \eps L - 1 \right) \right| \leq \left( 2 + 2c' \right) \eps L,
  \]
  and so by the triangle inequality we have that 
  $\ed \left( \wt{R} \left[ 1 + \alpha \eps L, -1 \right], X \left[ b_{i} + 1, b_{i+1} \right] \right) \leq \left( 3 + 2c' \right) \eps L$,
  proving~\eqref{eq:Y_i} for all $i \leq \tau_1$. The proof for $i \in \left[ \tau_1, \tau_2 \right]$ is similar, except now $b_i = b_{\tau_1}$ and $a_i$ changes. 

  We proved that for all $i \leq \tau_1$ we have 
  $\wt{T}_{i+1} - \wt{T}_i \geq \left( 1 - \left( \alpha + 2 + 2c' \right) \eps \right) L$. 
  A similar statement holds for $i \in \left[ \tau_1, \tau_2 \right]$, and together these imply part~\eqref{lem:analysis_tau} of the lemma.

  Since we have $a_{\tau} \leq L$ and $b_{\tau} \geq n - L$, this implies part~\eqref{lem:analysis_X_tau} of the lemma.   
\end{proof}

Putting everything together and using the triangle inequality we get that the algorithm outputs an estimate $\wh{X}$ which satisfies 
\[
 \ed \left( X, \wh{X} \right) \leq \frac{3+2c'}{ 1 - \left( \alpha + 2 + 2c' \right) \eps } \eps n + 2 L,
\]
which proves Theorems~\ref{thm:main_eps} and~\ref{thm:main_eps_3}.

\section{Discussion and future work} \label{sec:discussion} 

We introduced an adversarial error model for the problem of sequence assembly from shotgun reads. 
Our main result shows that if the reads are long enough and there is high enough coverage, 
then approximate reconstruction of the original sequence is possible for almost all sequences. 
The main question our work leaves open is: 
what are the fundamental information-theoretic limits to approximate reconstruction? 
Given $\overline{L}$ and $N$, is approximate reconstruction possible? 
If so, what is the best approximation factor achievable? 
What is the best ``strategy'' for an adversary that can corrupt the reads?

The probabilistic model we consider for the sequence of interest is simplistic, 
and it would be worthwhile to consider more general distributions, such as a Markov chain model. 
However, in many genomes there are long repeats, which are not captured by a Markov model. 
A direction for future research is to understand the fundamental limits to approximate reconstruction for arbitrary sequences as a function of their (approximate) repeat statistics.

Our adversarial error model also contains a simplification: 
sequencing technologies typically do not have a uniform error rate. 
Instead, while the error rate is reasonably small for most reads, 
there are some where the error rate is large and the resulting reads are useless.  
Practitioners can often detect these bad reads and thus throw them away. 
A variant of our algorithm can also handle very bad reads if the quality of good and very bad reads are sufficiently separated: 
the very bad reads simply will not align anywhere and so will be thrown out. 
However, if there is a continuous spectrum of quality from good to very bad reads, the algorithm runs into issues due to the reads in the middle of the spectrum. 
We leave addressing this issue as a future challenge.


\section*{Acknowledgements}

The research of E.M.\ is supported by 
NSF grant CCF-1320105, DOD ONR grant N00014-14-1-0823, and Simons Foundation grant 328025. 
M.Z.R.\ thanks Jasmine Nirody and Rachel Wang for helpful discussions.


\bibliographystyle{abbrv}
\bibliography{bib}



\appendix

\section{Proofs of edit distance results} \label{sec:edit_distance_proofs} 

\begin{proof}[Proof of Lemma~\ref{lem:indpt}]
 For any $m$ and $n$ we clearly have 
\begin{multline*}
  \ed \left( X_{m+n}, Y_{m+n} \right) \\ 
  \leq 
  \ed \left( X_{m+n} \left[ 1 , m \right], Y_{m+n} \left[ 1 , m \right] \right) 
  +
  \ed \left( X_{m+n} \left[ m + 1 , m + n \right], Y_{m+n} \left[ m + 1 , m + n \right] \right).
\end{multline*}
Thus Kingman's subadditive ergodic theorem implies that 
$\lim_{m \to \infty} \frac{1}{m} \ed \left( X_m, Y_m \right) =: c_{\ind}$ 
exists almost surely. 
Clearly $c_{\ind} \geq 0$; 
what remains to show is that $c_{\ind} > 0$. 
We do this via a volume argument; 
we first present a simple argument and then refine it to get a better lower bound on $c_{\ind}$.

If $\ed \left( X_m, Y_m \right) \leq r$, 
then one can get from $X_m$ to $Y_m$ using at most $r$ deletions, insertions and substitutions. 
The locations of the at most $r$ deletions and substitutions can be chosen in at most $\binom{m}{r}$ ways, 
and the same holds for the locations of the at most $r$ insertions and substitutions. 
Given the locations of these, there can be at most $\left| \Sigma \right|^r$ subsequences in these locations. 
That is, the edit distance ball of radius $r$ around any point $x \in \Sigma^m$ contains at most 
$\binom{m}{r}^2 \left| \Sigma \right|^{r}$ 
points of $\Sigma^m$. 
For $r = \delta m$ we get 
\[
  \binom{m}{\delta m}^2 \left| \Sigma \right|^{\delta m} 
  \approx 
  2^{2 H \left( \delta \right) m} \left| \Sigma \right|^{\delta m},
\]
where 
$H \left( x \right) = - x \log_2 \left( x \right) - \left( 1 - x \right) \log_2 \left( 1 - x \right)$ 
is the binary entropy function. 
Note that the total number of sequences of length $m$ is $\left| \Sigma \right|^m$.
Let $\delta^{*} = \delta^{*} \left( \Sigma \right)$ be the unique solution in $\left( 0, 1 \right)$ of 
$4^{H\left( \delta \right)} \left| \Sigma \right|^{\delta} = \left| \Sigma \right|$. 
By the volume argument above we have that 
for any $\delta < \delta^{*}$  
the probability that
$\ed \left( X_m, Y_m \right) \leq \delta m$ 
is exponentially small in $m$. 
Thus 
$c_{\ind} \geq \delta^{*} > 0$. 
In particular, 
for $\left| \Sigma \right| = 2$, we have $\delta^* \approx 0.09488$, and 
for $\left| \Sigma \right| = 4$, we have $\delta^* \approx 0.22709$.

We can obtain a better bound by a slightly more careful argument. 
Again, if $\ed \left( X_m, Y_m \right) \leq r$, 
then one can get from $X_m$ to $Y_m$ using at most $r$ deletions, insertions and substitutions. 
Suppose that 
the number of deletions is $D$, 
the number of insertions is $I$, and 
the number of substitutions is $S$. 
Since $X_m$ and $Y_m$ have the same length, 
we have $D = I$ 
and also $D + I + S \leq r$, i.e., $S \leq r - 2D$. 
The locations of the $D$ deletions can be chosen in at most $\binom{m}{D}$ ways, 
the locations of the $I$ insertions can be chosen in at most $\binom{m}{I}$ ways, 
while the locations of the $S$ substitutions can be chosen in at most $\binom{m}{S}$ ways. 
Given the locations of these, 
there can be at most $\left| \Sigma \right|^I$ subsequences in the locations of the insertions, 
and at most $\left| \Sigma \right|^S$ subsequences in the locations of the substitutions. 
Therefore the edit distance ball of radius $r$ around any point $x \in \Sigma^m$ contains at most 
\[
\max_{0 \leq D \leq r/2} \left\{ \binom{m}{D}^2 \binom{m}{r - 2D} \left| \Sigma \right|^{ r - D } \right\}
\]
points of $\Sigma^m$. 
For $r = \delta m$ and $D = \delta_D m$ we have
\[
  \binom{m}{\delta_D m}^2 \binom{m}{\left( \delta - 2\delta_D \right) m} \left| \Sigma \right|^{\left( \delta - \delta_D \right) m} 
  \approx 
  2^{\left( 2 H \left( \delta_D \right) + H \left( \delta - 2\delta_D \right)  - \delta_D \log_2 \left( \left| \Sigma \right| \right) \right) m} 
  \left| \Sigma \right|^{\delta m}. 
\]
Let $\delta^{**} = \delta^{**} \left( \Sigma \right)$ be the unique solution in $\left( 0, 1 \right)$ of 
\[
  2^{\max_{0 \leq x \leq \delta / 2} \left\{ 2 H \left( x \right) + H \left( \delta - 2x \right) - x \log_2 \left( \left| \Sigma \right| \right) \right\}} \left| \Sigma \right|^{\delta} 
  =
  \left| \Sigma \right|.
\]
The volume argument thus tells us that for every $\delta < \delta^{**}$ 
the probability that 
$\ed \left( X_m, Y_m \right) \leq \delta m$ 
is exponentially small in $m$. 
Thus 
$c_{\ind} \geq \delta^{**}$. 
In particular, 
for $\left| \Sigma \right| = 2$, we have $\delta^{**} \approx 0.15776$, and 
for $\left| \Sigma \right| = 4$, we have $\delta^{**} \approx 0.33832$.   
\end{proof}

\begin{proof}[Proof of Lemma~\ref{lem:small_overlap}]
By~\eqref{eq:ed_mon} we have that 
\[
\ed \left( X \left[ 1 , m \right], X \left[k + 1, k + m  \right] \right) 
\geq 
\ed \left( X \left[ 1 , dm \right], X \left[k + 1, k + dm  \right] \right).
\]
Since $k + 1 > dm$, the strings 
$X \left[ 1 , dm \right]$ 
and 
$X \left[ k + 1 , k + dm \right]$
are independent uniformly random strings of length $dm$. 
Recall the definition of $\delta^{**}$ from the proof of Lemma~\ref{lem:indpt} 
and let $\delta \in \left( 0, \delta^{**} \right)$. 
In the proof of Lemma~\ref{lem:indpt} we showed that 
the probability that 
$\ed \left( X \left[ 1 , dm \right], X \left[k + 1, k + dm  \right] \right) 
\leq \delta d m$
is exponentially small in $dm$. 
By taking a union bound over $k \in \left[ dm, m \right]$ 
we arrive at the desired result 
with, e.g., $\gamma \left( d, \Sigma \right) = 0.9 \times \delta^{**} \left( \Sigma \right) d$, 
and an appropriate constant $c'$. 
\end{proof}

Before proving Lemma~\ref{lem:large_overlap} we introduce a variant of the edit distance which is simpler to understand theoretically 
and for which we state and prove a result similar to Lemma~\ref{lem:large_overlap}. 
We denote by $\edd \left( x, y \right)$ the minimum number of deletion or insertion operations necessary to go from $x$ to $y$;  
that is, compared to the edit distance, substitutions are not allowed. 
Since a substitution can be simulated by a deletion followed by an insertion, we have that 
$\ed \left( x, y \right) \leq \edd \left( x, y \right) \leq 2 \ed \left( x, y \right)$ 
for all $x, y \in \Sigma^*$. 
The nice property of this distance is that 
$\edd \left( x, y \right) 
= \left| x \right| + \left| y \right| - 2 \LCS \left( x, y \right)$,
where $\LCS \left( x, y \right)$ is the length of the longest common subsequence (LCS) of $x$ and $y$. 
The following result is about the longest common subsequence of two random strings and is similar to Lemma~\ref{lem:large_overlap}. 

\begin{lemma}\label{lem:LCS}
Let $X \in \Sigma^{2m}$ be a uniformly random string. 
There exist positive constants $c = c \left( \Sigma \right)$ and $c' = c' \left( \Sigma \right)$ such that    
\[
    \LCS \left( X \left[ 1 , m \right], X \left[ 1 + k , m + k \right] \right) = m - k
\]
for all $k \le c m$ with probability at least $1 - e^{-c' m}$.
\end{lemma}

\begin{proof}
  It is immediate that 
  \[
   \LCS \left( X \left[ 1 , m \right], X \left[ 1 + k , m + k \right] \right) \geq m - k,
  \]
  since the last $m - k$ coordinates of $X\left[1,m\right]$ 
  and the first $m-k$ coordinates of $X \left[ 1 + k , m + k \right]$ are the same. 
  What remains is to show that the probability of
  \begin{equation}\label{eq:large_LCS}
  \LCS \left( X \left[ 1 , m \right], X \left[ 1 + k , m + k \right] \right) > m - k 
  \end{equation}
  is exponentially small in $m$. 
  Note that if a common subsequence of $X \left[ 1 , m \right]$ and $X \left[ 1 + k , m + k \right]$
  is such that the $\ell^{\text{th}}$ coordinate of $X \left[ 1 , m \right]$ 
  is mapped to the $\left( \ell - k \right)^{\text{th}}$ coordinate of $X \left[ 1 + k , m + k \right]$ (the ``trivial'' map), 
  then this subsequence can have length at most $m - k$, 
  since in $X \left[ 1, m \right]$ there are only $m - \ell$ coordinates to the right of this coordinate, 
  while in $X \left[ 1 + k , m + k \right]$ there are only $\ell - k - 1$ coordinates to the left of this coordinate. 
  So if a common subsequence has length greater than $m - k$, then every coordinate is mapped ``nontrivially''. 
  This then creates many constraints on the pair of sequences and so there will not be many of them, as we now argue.

  Suppose that 
  $\LCS \left( X \left[ 1 , m \right], X \left[ 1 + k , m + k \right] \right) = m - \ell$
  for some $\ell \in \left\{ 0, 1, \dots, k - 1 \right\}$. 
  A noncrossing matching between $m - \ell$ coordinates of the two sequences 
  is characterized by the $\ell$ coordinates in each sequence that are not part of the matching; 
  the remaining pairs of the matching are determined due to the noncrossing property. 
  Thus there are $\binom{m}{\ell}^2$ such matchings. 
  The coordinates of a LCS between the two sequences corresponds to a noncrossing matching between the two, 
  and it imposes conditions on the values of these coordinates. 
  If 
  $\LCS \left( X \left[ 1 , m \right], X \left[ 1 + k , m + k \right] \right) = m - \ell$
  then all but $\ell$ coordinates of $X \left[ 1 + k , m + k \right]$ are determined by $X \left[ 1 , m \right]$. 
  Furthermore, there are at least $m - k - \ell$ coordinates $j \in \left[ 1 + k, m \right]$ such that 
  the $\left( j - k \right)^{\text{th}}$ coordinate of $X \left[ 1 + k , m + k \right]$ is in the matching, 
  and thus the value of $X \left[ j \right]$ is determined by a previous value $X \left[ i \right]$ for some $i < j$. 
  This means that at most $k + \ell$ coordinates of $X \left[ 1, m \right]$ are not determined by the value of a previous coordinate. 
  So the probability of any given noncrossing matching being a common subsequence is at most 
  $\left| \Sigma \right|^{\left( k + 2\ell \right) - \left( m + k \right)} = \left| \Sigma \right|^{2\ell - m}$. 
  We have thus shown that the probability of~\eqref{eq:large_LCS} is at most 
  \[
   \sum_{\ell=0}^{k-1} \binom{m}{\ell}^2 \left| \Sigma \right|^{2\ell - m} \leq k \binom{m}{k}^2 \left| \Sigma \right|^{2k - m}.
  \]
  When $k = dm$, then this is approximately 
  $\left( d m \right) \times 2^{\left\{ 2 H \left( d \right) + \left( 2d - 1 \right) \log_2 \left( \left| \Sigma \right| \right) \right\} m}$,
  which goes to zero exponentially in $m$ if $d \in \left( 0, 1 \right)$ is small enough. 
  Taking a union bound over $k$ we have that this holds for all $k \in \left\{ 0, 1, \dots,  dm \right\}$ simultaneously. 
\end{proof}

\begin{proof}[Proof of Lemma~\ref{lem:large_overlap}] 
The proof is very similar to the one above. 
It is again immediate that 
$\ed \left( X \left[ 1, m \right], X \left[ 1 + k, m + k \right] \right) \leq 2k$, 
since one can obtain $X \left[ 1 + k, m + k \right]$ from $X \left[ 1, m \right]$
by first deleting the first $k$ coordinates of $X\left[1,m\right]$ 
and then inserting $X \left[ m + 1 , m + k \right]$ at the end of the sequence. 
What remains is to show that the probability of
\begin{equation}\label{eq:large_ed}
  \ed \left( X \left[ 1 , m \right], X \left[ 1 + k , m + k \right] \right) < 2k 
\end{equation}
is exponentially small in $m$. 
If~\eqref{eq:large_ed} holds then one can get 
from $X \left[ 1 , m \right]$ to $X \left[ 1 + k , m + k \right]$ 
by first performing $S$ substitutions to get $X' \left[ 1, m \right]$, 
then performing $D$ deletions and finally $I$ insertions. 
These quantities have to satisfy $S + D + I \leq 2k - 1$ and $D = I$. 
We thus have that 
$\edd \left( X' \left[ 1, m \right], X \left[ 1 + k, m + k \right] \right) \leq 2k - 1 - S$, 
and so 
$\LCS \left( X' \left[ 1, m \right], X \left[ 1 + k, m + k \right] \right) \geq m - k + \left( 1 + S \right) / 2$. 
Let $m - \ell := \LCS \left( X' \left[ 1, m \right], X \left[ 1 + k, m + k \right] \right)$. 
Again, the number of such noncrossing matchings is $\binom{m}{\ell}^2$. 
As in the proof of the previous lemma, 
the noncrossing matching corresponding to such a LCS has to map every coordinate ``nontrivially''. 
Every such matching imposes constraints on $X' \left[ 1, m \right]$ and $X \left[ 1 + k, m + k \right]$, 
and while there are less constraints as before---since $X' \left[ 1, m \right]$ differs from $X \left[ 1, m \right]$ in $S$ substitutions---we can still show that the probability of~\eqref{eq:large_ed} is at most
\[
 k \binom{m}{2k} \binom{m}{k}^2 \left| \Sigma \right|^{3k - m}.
\]
When $k = dm$, then this is approximately
$\left( d m \right) \times 2^{\left\{ 2 H \left( d \right) + H \left( 2 d \right) + \left( 3d - 1 \right) \log_2 \left( \left| \Sigma \right| \right) \right\} m}$,
which goes to zero exponentially in $m$ if $d \in \left( 0, 1 \right)$ is small enough; 
in particular this happens when $d \leq 0.0846092$. 
Taking a union bound over $k$ we have that this holds for all $k \in \left\{ 0, 1, \dots,  dm \right\}$ simultaneously. 
\end{proof}

\end{document}